\documentclass[sigconf]{acmart}

\usepackage{booktabs}
\usepackage{xcolor}
\usepackage{soul}
\usepackage[utf8]{inputenc}
\usepackage{amsfonts}
\usepackage{amsmath}
\usepackage{amssymb}
\usepackage{amsthm} 
\usepackage{graphicx}
\usepackage{float}
\usepackage{multirow}
\usepackage[ruled]{algorithm2e}
\usepackage{latexsym}
        
\newtheorem{theorem}{Theorem}  
\newtheorem{definition}{Definition}  
  
\newtheorem{corollary}{Corollary}  
\newtheorem{assumption}{Assumption}
\newcommand{\hide}[1]{} 
\newcommand{\vpara}[1]{\vspace{0.1in}\noindent\textbf{#1 }}

\newcommand{\beq}[1]{\begin{equation}#1\end{equation}
	\normalsize
}

\newcommand{\model}{{{G}{raph}{\large\emph{S}}{\footnotesize GAN}}}
\newcommand{\smodel}{\model\space}

\copyrightyear{2018}
\acmYear{2018}
\setcopyright{acmcopyright}
\acmConference[CIKM '18]{The 27th ACM International Conference on Information and Knowledge Management}{October 22--26, 2018}{Torino, Italy}
\acmBooktitle{The 27th ACM International Conference on Information and Knowledge Management (CIKM '18), October 22--26, 2018, Torino, Italy} \acmPrice{15.00}
\acmDOI{10.1145/3269206.3271768} \acmISBN{978-1-4503-6014-2/18/10}

\begin{document}
\title{
Semi-supervised Learning on Graphs with Generative Adversarial Nets
}


 \author{Ming Ding}
 \affiliation{%
   \institution{Tsinghua University}
   \city{Beijing, China}
 }
 \email{dm18@mails.tsinghua.edu.cn}

 \author{Jie Tang}
 \affiliation{%
   \institution{Tsinghua University}
   \city{Beijing, China}
 }
 \email{jietang@tsinghua.edu.cn}

 \author{Jie Zhang}
 \affiliation{%
   \institution{Tsinghua University}
   \city{Beijing, China}
 }
 \email{j-z16@mails.tsinghua.edu.cn}
\renewcommand{\shortauthors}{}
\begin{abstract}
	We investigate how generative adversarial nets (GANs) can help semi-supervised learning on graphs.
	We first provide insights on working principles of adversarial learning over graphs and then
	present \model, a novel approach to semi-supervised learning on graphs.
    In \model, generator and classifier networks play a novel competitive game. At equilibrium, generator generates fake samples in low-density areas between subgraphs. In order to discriminate fake samples from the real, classifier implicitly takes the density property of subgraph into consideration. An efficient adversarial learning algorithm has been developed to improve traditional normalized graph Laplacian regularization with a theoretical guarantee.
    
	Experimental results on several different genres of datasets show that the proposed \smodel significantly outperforms several state-of-the-art methods.
	\smodel can be also trained using mini-batch, thus enjoys the scalability advantage.

\end{abstract}

%
%



\maketitle
\section{Introduction}\label{sec:intro}

\hide{
Many learning-based applications on graphs suffer from statistically insufficient labeled data. 
For instance, in a large citation graph, it would be difficult to have a large number of labeled papers on ``SVM applications''.
Semi-supervised learning thus has been proposed to improve the learning accuracy by leveraging the large unlabeled data for training.
}

Semi-supervised learning on graphs has attracted great attention  both in theory and practice. 
Its basic setting is that we
are given a graph comprised of a small set of labeled nodes and a large set of unlabeled nodes, and the goal is to learn a model that can predict label of the unlabeled nodes.

There is a long line of works about semi-supervised learning over graphs. 
One important category of the research is mainly based on the graph Laplacian regularization framework. 
For example, 
Zhu et al.~\shortcite{zhu2002learning} proposed a method called Label Propagation for learning from labeled and unlabeled data on graphs, and later the method has been improved by Lu and Getoor~\shortcite{lu2003link} under the bootstrap-iteration framework.
Blum and Chawla~\shortcite{blum2001learning} also formulated the graph learning problem  as that of finding min-cut on graphs. 
Zhu et al.~\shortcite{zhu2003semi} proposed an algorithm based on Gaussian random field and formalized graph Laplacian regularization framework. 
Belkin et al.~\shortcite{belkin2006manifold} presented a 
regularization method called ManiReg by exploiting geometry of marginal distribution for semi-supervised learning. 
The second category of the research is to combine semi-supervised learning with graph embedding.
Weston et al.~\shortcite{weston2012deep} first incorporated deep neural networks into the graph Laplacian regularization framework for semi-supervised learning and embedding. Yang et al.~\shortcite{yang2016revisiting} proposed the Planetoid model for jointly learning graph embedding and predicting node labels. Recently, Defferrard et al.~\shortcite{defferrard2016convolutional} utilized localized spectral Chebyshev filters to perform convolution on graphs for machine learning tasks. Graph convolution networks (GCN) \cite{kipf2016semi} and its extension based on attention techniques~\cite{DBLP:journals/corr/abs-1710-10903} demonstrated great power and achieved state-of-art performance on this problem.  

This paper investigates the potential of generative adversarial nets (GANs) for  semi-supervised learning over graphs.
GANs~\cite{goodfellow2014generative} are originally designed for generating images, by training two neural networks which play a min-max game: discriminator $D$ tries to discriminate real from fake samples and generator $G$ tries to generate ``real'' samples to fool the discriminator. 
To the best of our knowledge, there are few works on semi-supervised learning over graphs with GANs.

We present a novel method \smodel for semi-supervised learning on graphs with GANs.
\smodel
maps graph topologies into feature space and jointly trains generator network and classifier network. 
Previous works~\cite{dai2017good,kumar2017semi} tried to explain semi-supervised GANs' working principles, but only found that generating {\it moderate fake samples} in complementary areas benefited classification and analyzed under strong assumptions. 
This paper explains the working principles behind the proposed model from the perspective of game theory.
We have an intriguing observation that 
fake samples in low-density areas between subgraphs
can reduce the influence of samples nearby, thus help improve the classification accuracy. 
A novel GAN-like game is designed under the guidance of this observation. Sophisticated losses guarantee the generator generates samples in these low-density areas at equilibrium. 
In addition, integrating with the observation, the
graph Laplacian regularization framework (Equation (\ref{basic})) can leverage clustering property to make stable progress.
It can be theoretically proved that this adversarial learning technique yields perfect classification for semi-supervised learning on graphs with plentiful but finite generated samples. 

The proposed \smodel is evaluated on several different genres of datasets. Experimental results show that \smodel significantly outperforms several state-of-the-art methods. \smodel can be also trained using mini-batch, thus enjoys the scalability advantage. 

Our contributions are as follows:
\begin{itemize}
\item We introduce GANs as a tool to solve classification tasks on graphs under semi-supervised settings. \smodel generates fake samples in low-density areas in graph and leverages clustering property to help classification. 

\item We formulate a novel competitive game between generator and discriminator for \smodel and thoroughly analyze the dynamics, equilibrium and working principles during training. In addition, we generalize the working principles to improve traditional algorithms. Our theoretical proof and experimental verification both outline the effectiveness of this method.  

\item We evaluate our model on several dataset with different scales. \smodel significantly outperforms previous works and demonstrates outstanding scalability.
\end{itemize}

The rest of the paper is arranged as follows. In Section~\ref{sec:def}, we introduce the necessary definitions and GANs. In Section~\ref{sec:gen_fs}, we present \smodel and discuss why and how the model is designed in detail. A theoretical analysis of the working principles behind \model is given in Section ~\ref{sec:theory}. We outline our experiments in Section~\ref{sec:experiments} and show the superiority of our model. We close with a summary of related work in Section~\ref{sec:related}, and our conclusions.

\section{Preliminaries}
\label{sec:def}
\subsection{Problem Definition}
    Let $G=(V,E)$ denote a graph, where $V$ is a set of nodes and $E \subseteq V \times V $ is a set of edges. Assume each node $v_i$ is associated with a $k-$dimensional real-valued feature vector $\mathbf{w}_i \in \mathbb{R}^k$ and a label $y_i \in \{0,...,M-1\}$.
    If the label $y_i$ of node $v_i$ is unknown, we say node $v_i$ is an unlabeled node. We denote the set of labeled nodes as $V^L$ and the set of unlabeled nodes as $V^U=V\backslash V^L$. Usually, we have $|V^L| \ll |V^U|$. We also call the graph $G$ as \textit{partially labeled graph}~\cite{Tang:11PKDD}.
    Given this, we can formally define the semi-supervised learning problem on graph.

\begin{definition}\textbf{Semi-supervised Learning on Graph.}
	Given a partially labeled graph $G=(V^L\cup V^U, E)$, the objective here is to learn a function $f$ using features $\mathbf{w}$ associated with each node and the graphical structure, in order to predict the labels of unlabeled nodes in the graph.
\end{definition}
 
 Please note that in semi-supervised learning, training and prediction are usually performed simultaneously. In this case, the learning considers both labeled nodes and unlabeled nodes, as well as the structure of the whole graph.
    In this paper, we mainly consider transductive learning setting, though the proposed model can be also applied to other machine learning settings. 
    Moreover, we only consider undirected graphs, but the extension to directed graphs is straightforward.

\subsection{Generative Adversarial Nets (GANs)}

GAN~\cite{goodfellow2014generative} is a new framework for estimating generative models via an adversarial process, in which a generative model $G$ is trained to best fit the original training data and a discriminative model $D$ is trained to distinguish real samples from samples generated by model $G$.
The process can be formalized as a min-max game between $G$ and $D$, with the following loss (value) function:

\beq{\label{gan_aim}
\min\limits_G\max\limits_D V(G,D) = \mathbb{E}_{\mathbf{x}\sim p_{d}(\mathbf{x})} \log D(\mathbf{x}) 
+ \mathbb{E}_{\mathbf{z}\sim p_z(\mathbf{z})} \log [1-D(G(\mathbf{z}))]
}

\noindent where $p_{d}$ is the data distribution from the training data, $p_z(\mathbf{z})$ is a prior on input noise variables.

\begin{figure}[t]
	\centering
	\includegraphics[width=\linewidth]{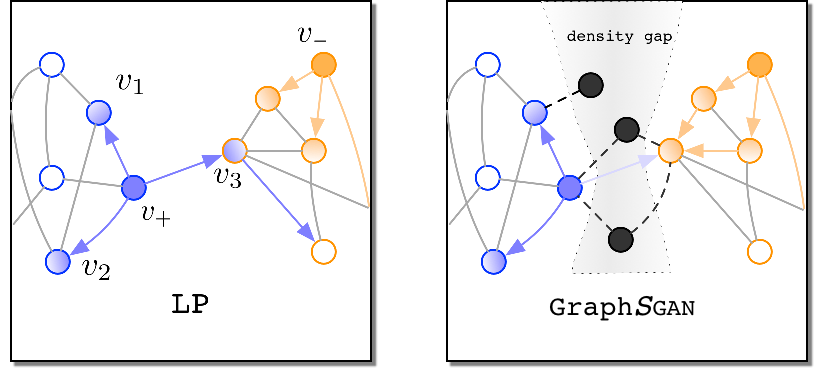}
	\caption{A qualitative illustration of  working principles in \model. 
		The two labeled nodes are in solid blue and solid orange respectively and all the other nodes are unlabeled nodes.	
		The left figure is a bad case of vanilla Label Propagation algorithm. In this case, node $v_3$ is assigned a wrong label due to its direct link to  node $v_+$. 
		The right figure illustrates how \smodel works. It generates fake nodes (in black) in the density gap thus reduces the influence of nodes across the density gap.
}\label{example}
\end{figure}

 \begin{figure*}
 \centering
 \includegraphics[width=0.75\linewidth]{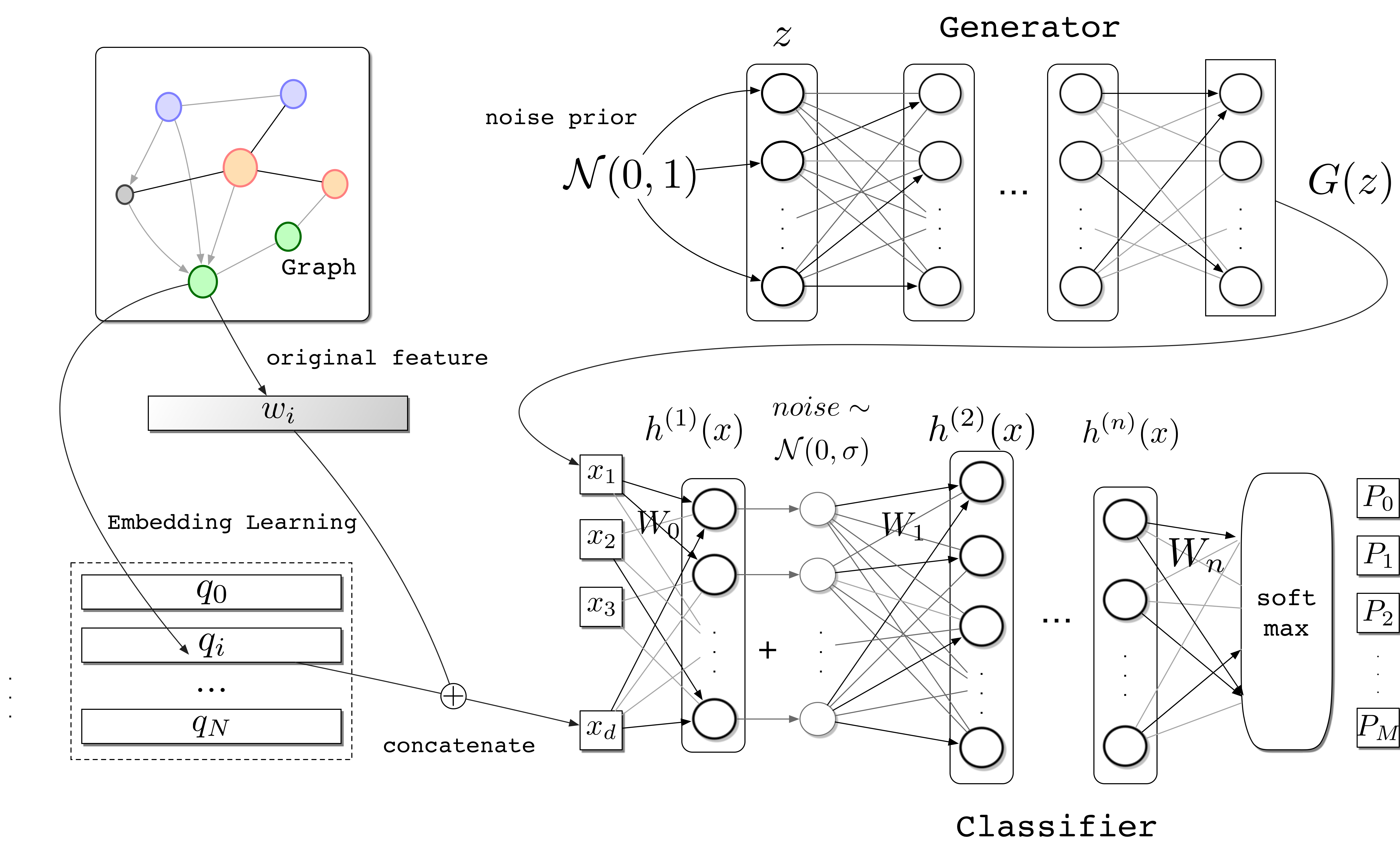}
 \caption{An overview of our model. Fake inputs are generated by generator and real inputs are acquired by concatenating original feature $\mathbf{w}_i$ and learned embedding $\mathbf{q}_i$. Both real inputs and fake samples generated by generator are fed into the classifier.}\label{overview}
 \end{figure*}
 
\section{Model Framework}
\label{sec:gen_fs}

\subsection{Motivation}
\label{sec:mot}
We now introduce how we leverage the power of GANs for semi-supervised learning over graphs. Directly applying GAN to graph learning is infeasible, as it does not consider the graph structure. To show how GANs help semi-supervised learning over graphs, we begin with one example. The left figure in Figure~\ref{example} shows a typical example in graph-based semi-supervised learning. The two labeled nodes are in blue and orange respectively. Traditional methods such as Label Propagation~\cite{zhu2002learning} does not consider the graph topology, thus cannot differentiate the propagations from node $v_+$ to nodes $v_1$, $v_2$, and $v_3$. 
Taking a closer look at the graph structure, we can see there are two subgraphs. We call the area between the two subgraphs as \textbf{density gap}.

Our idea is to use GAN to estimate the density subgraphs and then generate samples in the density gap area. We then request the classifier firstly to discriminate fake samples before classifying them into different classes.
In this way, discriminating \textit{fake} samples from \textit{real} samples will result in a higher curvature of the learned classification function around density gaps, which 
weakens the effect of propagation across density gaps (as shown in the right figure of Figure~\ref{example}).
Meanwhile inside each subgraph, confidence on correct labels will be gradually boosted because of supervised loss decreasing and general smoothing techniques for example stochastic layer. A more detailed analysis will be reported in \S~\ref{sec:verify}.

\subsection{Architecture}
\label{arch}
 	GAN-based models cannot be directly applied to graph data.
 	To this end, \smodel first uses network embedding methods (e.g., DeepWalk~\cite{Perozzi:14KDD}, LINE~\cite{tang2015line}, or NetMF~\cite{Qiu:2018WSDM}) to learn latent distributed representation $\mathbf{q}_i$ for each node, and then concatenates the latent distribution $\mathbf{q}_i$ with the original feature vector $\mathbf{w}_i$, i.e., $\mathbf{x}_i=(\mathbf{w}_i, \mathbf{q}_i)$. 
 	Finally, 
 	$\mathbf{x}_i$ is taken as input to our method.
 	
 	Figure~\ref{overview} shows the architecture of \model.
	Both classifier $D$ and generator $G$ in \smodel are multiple layer perceptrons.
	More specifically, the generator takes a Gaussian noise $\mathbf{z}$ as input and outputs fake samples having the similar shape as $\mathbf{x}_i$. In the generator, {\it batch normalization} ~\cite{ioffe2015batch} is used. 
 	Generator's output layer is constrained by {\it weight normalization} trick~\cite{salimans2016weight}  with a trainable weight scale.
    Discriminator in GANs is substituted by a classifier, where stochastic layers(additive Gaussian noise) are added after input and full-connected layers for smoothing purpose. Noise is removed in prediction mode. Parameters in full-connected layers are constrained by weight normalization for regularization. 
    Outputs of the last hidden layer in classifier $h^{(n)}(\mathbf{x})$\label{feature} are features extracted by non-linear transformation from input $\mathbf{x}$, which is essential for feature matching~\cite{salimans2016improved} when training generator. The classifier ends with a $(M+1)$-unit output layer and softmax activation. The outputs of unit $0$ to unit $M-1$ can be explained as probabilities of different classes and output of unit $M$ represents probability to be fake. In practice, we only consider the first $M$ units and assume the output for fake class $P_M$ is always 0 before softmax, because subtracting an identical number from all units before softmax does not change the softmax results.

\subsection{Learning Algorithm}
\subsubsection{Game and Equilibrium}
GANs try to generate samples similar to training data but we want to generate fake samples in density gaps. So, the optimization target must be different from original GANs in the proposed \smodel model. For better explanation, we revisit GANs from a more general perspective in game theory.

In a normal two-player game, $G$ and $D$ have their own loss functions and try to minimize them. Their losses are interdependent. We denote the loss functions $\mathcal{L}_G(G, D)$ and $\mathcal{L}_D(G, D)$. Utility functions $V_G(G,D)$ and $V_D(G, D)$ are negative loss functions.
	
GANs define a \emph{zero-sum game}, where $\mathcal{L}_G(G, D) = -\mathcal{L}_D(G, D)$. In this case, the only Nash Equilibrium can be reached by minimax strategies~\cite{von2007theory}. To find the equilibrium is equivalent to solve the optimization:
\beq{\nonumber\min\limits_G\max\limits_D V_D(G,D)}
    
Goodfellow et al. ~\cite{goodfellow2014generative} proved that if $V_D(G,D)$ was defined as that in equation \ref{gan_aim}, $G$ would generate samples subject to data distribution at the equilibrium. The distribution of generate samples $p_g(\mathbf{x})$ is an approximation of the distribution of real data $p_d(\mathbf{x})$. But we want to generate samples in density gaps instead of barely mimicking real data. So, original GANs cannot solve this task.

\begin{figure}[t]
\includegraphics[width=0.75\linewidth]{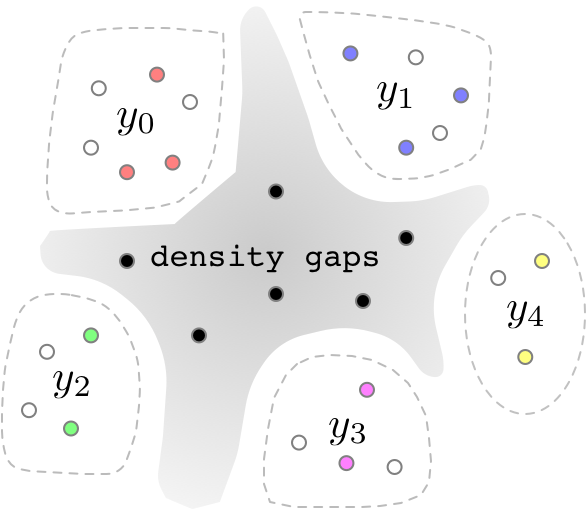}
\caption{An illustration of the expected equilibrium in $h^{(n)}(\mathbf{x})$. The dotted areas are clusters of training data with different labels(colorful points). Unlabeled samples(white points) are mapped into a particular clusters. Distinct density gaps appear in the center, in which lie the generated samples(black points).}\label{expected}
\end{figure}

In the proposed \model, we modify $\mathcal{L}_D(G, D)$ and $\mathcal{L}_G(G, D)$ to design a new game, in which $G$ would generate samples in density gaps at the equilibrium. More precisely, we expect that the real and fake samples are mapped like Figure \ref{expected} in its final representative layer $h^{(n)}(\mathbf{x})$. Because the concept of ``density gap'' is more straightforward in a representative layer than in a graph, we define that a node lies in density gap if and only if it lies in density gap in $h^{(n)}(\mathbf{x})$ layer. How to map nodes into representative layer is explained in section \ref{arch}.

The intuition behind the design is based on the famous phenomenon known as ``curse of dimensionality'' ~\cite{donoho2000high}. In high-dimensional space like $h^{(n)}(\mathbf{x})$, the central area is far more narrow than outer areas. Training data in the central area are easy to become \emph{hubs}~\cite{radovanovic2010hubs}. Hubs frequently occurs in the nearest neighbors of samples from other classes, which might deeply affect semi-supervised learning and become a main difficulty. So, we want the central area become a density gap instead of one of clusters.

We define $\mathcal{L}_D(G, D)$ and $\mathcal{L}_G(G, D)$ as below to guarantee the expected equilibrium.

\beq{\begin{aligned}
\mathcal{L}_D =& loss_{sup} + \lambda_0 loss_{un} + 
\lambda_1 loss_{ent} + loss_{pt}\\
\mathcal{L}_G =& loss_{fm} + \lambda_2 loss_{pt}
\end{aligned}
}

\noindent Next, we explain these loss terms in $\mathcal{L}_D(G, D)$ and $\mathcal{L}_G(G, D)$ and how they take effect in details.

\subsubsection{Discriminative Losses}\label{LD} 
At equilibrium, no player can change their strategy to reduce his loss unilaterally. Supposing that $G$ generates samples in central areas at equilibrium, we put forward four conditions for $D$ to guarantee the expected equilibrium in $h^{(n)}(\mathbf{x})$:

\begin{enumerate}
\item Nodes from different classes should be mapped into different clusters.
\item Both labeled and unlabeled nodes should \textbf{not} be mapped into the central area so as to make it a density gap.
\item Every unlabeled node should be mapped into one cluster representing a particular label.
\item Different clusters should be far away enough.
\end{enumerate}

The most natural way to satisfy condition (1) is a supervised loss $loss_{sup}$. $loss_{sup}$ is defined as the cross entropy between predicted distribution over $M$ classes and one-hot representation for real label. 
    
\begin{equation}
	loss_{sup} = -\mathbb{E}_{\mathbf{x}_i\in X^L} \log P(y_i|\mathbf{x}_i, y_i < M)
\end{equation}

\noindent where $X^L$ is the set of inputs for labeled nodes $V^L$.
	
Condition (2) is equivalent to the $D$'s aim in original GAN given that $G$ generates fake samples in central density gaps. Thus we still use the loss in equation \ref{gan_aim} and call it $loss_{un}$. The classifier $D$ incurs $loss_{un}$ when real-or-fake misclassification happens.  
    
\begin{equation}
\begin{aligned}
	loss_{un} =& -\mathbb{E}_{\mathbf{x}_i \in X^U} \log [1-P(M|\mathbf{x}_i)] \\&- \mathbb{E}_{\mathbf{x}_i \sim G(\mathbf{z})} \log P(M|\mathbf{x}_i)
\end{aligned}
\end{equation}

\noindent where $X^U$ is set of pretreated inputs for unlabeled nodes $V^U$; $G(\mathbf{z})$ is the distribution of generated samples; and $P(M|\mathbf{x}_i)$ denotes the predicted fake probability of $\mathbf{x}_i$. 

Condition (3) requests $D$ to assign an unambiguous label to every unlabeled node. We solve the problem by adding an entropy regularization term $loss_{ent}$, the entropy of distribution over $M$ labels. Entropy is a measurement of uncertainty of probability distributions. It has become a regularization term in semi-supervised learning for a long time~\cite{grandvalet2005semi} and is firstly combined with GANs in ~\cite{springenberg2015unsupervised}. Reducing entropy could encourage the classifier to determine a definite label for every node.

\begin{equation}
	loss_{ent} = -\mathbb{E}_{\mathbf{x}_i \in X^U} \sum\limits_{y=0}^{M-1} P(y|\mathbf{x}_i, y_i < M)\log P(y|\mathbf{x}_i, y_i < M) 
\end{equation}

Condition (4) widens density gaps to help classification. We leverage pull-away term $loss_{pt}$~\cite{zhao2016energy} to satisfy it. $loss_{pt}$ is originally designed for generating diverse samples in ordinary GANs. It is the average cosine distance between vectors in a batch. It keeps representations in $h^{(n)}(\mathbf{x})$ layer as far from the others as possible. Hence, it also encourages clusters to be far from the others.

    \begin{equation}\label{loss_pt}
    loss_{pt} = \frac{1}{m(m-1)}\sum\limits_{i=1}^m\sum\limits_{j\neq i}\frac{h^{(n)}(\mathbf{x}_i)^\top h^{(n)}(\mathbf{x}_j)}{||h^{(n)}(\mathbf{x}_i)||||h^{(n)}(\mathbf{x}_j)||}^2
\end{equation}

\noindent where $\mathbf{x}_i,\mathbf{x}_j$ are in the same batch and $m$ is batch size.

\subsubsection{Generative Losses} Similarly, supposing that $D$ has satisfied the four conditions above, we also have two conditions for $G$ to guarantee the expected equilibrium in $h^{(n)}(\mathbf{x})$:

\begin{enumerate}
\item $G$ generates samples which are mapped into the central area.
\item Generated samples should \textbf{not} overfit at the only center point.
\end{enumerate}

For condition (1), we train $G$ using \emph{feature matching} loss~\cite{salimans2016improved}. It minimizes the distances between generated samples and the center point of real samples $\mathbb{E}_{\mathbf{x}_i \in X^U\cup X^L} h^{(n)}(\mathbf{x}_i)$. Actually, in training process the center point is replaced by center of samples in a real batch $\mathbb{E}_{\mathbf{x}_i \in X_{batch}} h^{(n)}(\mathbf{x}_i)$, which helps satisfy condition (2).
The distances are originally measured in $L2$ norm. (But, in practice, we found that $L1$ norm also works well, with even slightly better performance.)

\begin{equation}
	loss_{fm} = ||\mathbb{E}_{\mathbf{x}_i \in X_{batch}} h^{(n)}(\mathbf{x}_i) - \mathbb{E}_{\mathbf{x}_j \sim G(\mathbf{z})} h^{(n)}(\mathbf{x}_j)||_2^2
\end{equation}

Condition (2) requests generated samples to cover as much central areas as possible. We also use a pull-away loss term(Equation \ref{loss_pt}) to guarantee the satisfication of this condition, because it encourage $G$ to generate diverse samples. A trade-off is needed between centrality and diversity, thus we use a hyper-parameter $\lambda_2$ to balance $loss_{fm}$ and $loss_{pt}$. The stochastic layers in $D$ add noise to fake inputs, which not only improves robustness but also prevents fake samples from overfitting.

\begin{algorithm}[t]
\caption{Minibatch stochastic gradient descent training of \model}\label{algo}
\KwIn{
Node features $\{\mathbf{w}_i\}$, Labels $y^L$, Graph $G=(V,E)$, Embedding Algorithm $\mathcal{A}$, batch size $m$.}
Calculate $\{\mathbf{w}_i'\}$ according to Eq. (\ref{neighbor_fusion}) \\
Calculate $\{\mathbf{q}_i\}$ via $\mathcal{A}$\\
Concatenate $\{\mathbf{w}_i'\}$ with $\{\mathbf{q}_i\}$ for $X^L\cup X^U$ \\
\Repeat{Convergence}{
	Sample $m$ labeled samples $\{\mathbf{x}^L_1,..., \mathbf{x}^L_m\}$ from $X^L$\\
    Sample $m$ unlabeled samples $\{\mathbf{x}^U_1,..., \mathbf{x}^U_m\}$ from $X^U$\\
    Sample $m$ noise samples $\{\mathbf{z}_1,...,\mathbf{z}_m\}$ from $p_\mathbf{z}(\mathbf{z})$\\
    Update the classifier by descending gradients of losses:
    $$\nabla_{\theta_D} \frac{1}{m}\sum loss_{sup} + \lambda_0 loss_{un} + \lambda_1 loss_{ent} + loss_{pt}$$ 
    \For{t steps}{
		Sample $m$ unlabeled samples $\{\mathbf{x}^U_1,..., \mathbf{x}^U_m\}$ from $X^U$\\
    	Sample $m$ noise samples $\{\mathbf{z}_1,...,\mathbf{z}_m\}$ from $p_\mathbf{z}(\mathbf{z})$\\
        Update the generator by descending gradients of losses:
        $$\nabla_{\theta_G} \frac{1}{m} \sum loss_{fm} + \lambda_2 loss_{pt}$$
    }
}
\end{algorithm}

\subsubsection{Training} GANs train $D$ and $G$ by iteratively minimizing $D$ and $G$'s losses. In game theory, it is called \emph{myopic best response}~\cite{aumann1974subjectivity}, an effective heuristic method to find equilibriums. \smodel is also trained in this way. 

The first part of training is to turn nodes in the graph to vectors in feature space. We use LINE~\cite{tang2015line} for pretreatment of $\mathbf{q}$, which performs fast and stable on our dataset. We also test other network embedding algorithms and find similar performances in classification. To accelerate the convergence, nodes' features $\mathbf{w}$ are recalculated using \textit{neighbor fusion} technique. Let $Ne(v_i)$ be the set of neighbors of $v_i$, node $v_i$'s weights are recalculated by

\begin{equation}\label{neighbor_fusion}
\mathbf{w}_i' = \alpha \mathbf{w}_i + \frac{1-\alpha}{|Ne(v_i)|}\sum\limits_{v_j \in Ne(v_i)} \mathbf{w}_j.
\end{equation}

The neighbor fusion idea is similar to the pretreatment tricks using attention mechanisms
~\cite{DBLP:journals/corr/abs-1710-10903}.

In the main training loop, we iteratively trains $D$ and $G$. 
To compute $\mathcal{L}_D$, we need three batches of labeled, unlabeled and generated samples respectively. $loss_{sup}$ needs labeled data. $loss_{un}$ is computed based on unlabeled and generated data. Theoretically, $loss_{un}$ should also take account of labeled data to make sure that they are classified as real. But $loss_{sup}$ has made labeled data classified correctly as its real label so that it is not necessary to consider labeled data in $loss_{un}$. $loss_{ent}$ only considers unlabeled data and $loss_{pt}$ should ``pull away'' both labeled and unlabeled data. 
Usually three hyperparameters are needed to balance the scales of four losses. We only use two parameters $\lambda_0$ and $\lambda_1$ in \smodel because $loss_{sup}$ will soon be optimized nearly to 0 due to few labeled samples under the semi-supervised setting.

Both real and generated batches of data are needed to train $G$. $loss_{fm}$ compares the batch center of real and generated data in $h^{(n)}(\mathbf{x})$ and $loss_{pt}$ measures the diversity of generated data. We always want $G$ to generate samples in central areas, which is an assumption when discussing $\mathcal{L}_D$ in section ~\ref{LD}. So, we train $G$ several steps to convergence after every time we train $D$. 
Detailed process is illustrated in Algorithm \ref{algo}.

\section{Theoretical Basis}\label{sec:theory}
We provide theoretical analyses on why GANs can help semi-supervised learning on graph. In section \ref{sec:mot}, we claim that the working principle is to reduce the influence of labeled nodes across density gaps. In view of the difficulty to directly analyze the dynamics in training of deep neural networks,  
we base the analysis on the graph Laplacian regularization framework. 
\begin{definition}\label{marginal}\textbf{Marginal Node and Interior Node.}
Marginal Nodes $\mathcal{M}$ are nodes linked to nodes with different labels while Interior Nodes $\mathcal{I}$ not. Formally, $\mathcal{M} = \{v_i | v_i \in V \land (\exists v_j \in V, (v_i, v_j) \in E \land y_i \neq y_j)\}$, $\mathcal{I} = V \setminus \mathcal{M}$.
\end{definition}

\begin{assumption} \label{convergence}\textbf{Convergence conditions.} 
When $G$ converges, we expect it to generate fake samples linked to nearby marginal nodes. More specifically, let $V_g$ and $E_g$ be the set of generated fake samples and generated links from generated  nodes to nearby original nodes. we have  $\forall v_g \in V_g, (\exists v_i \in \mathcal{M}, (v_g, v_i) \in E_g) \land (\forall (v_g, v_i) \in E_g, v_i \in \mathcal{M}) $.  
\end{assumption}
\vspace{0.06in}

The loss function of graph Laplacian regularization framework is as follows:

\begin{equation}\label{basic}
	\mathcal{L}(y') = \sum\limits_{v_i\in V^L} loss(y_i, y_i') + \lambda\sum\limits_{v_i,v_j \in V}^{i\neq j} \alpha_{ij}\cdot neq(y_i', y_j')
\end{equation}

\noindent where $y_i'$ denotes predicted label of node $v_i$. The $loss(\cdot, \cdot)$ function measures the supervised loss between real and predicted labels. $neq(\cdot, \cdot)$ is a  0-or-1 function representing {\it not equal}.

\begin{equation}\label{alpha}
\alpha_{ij} = \tilde{A_{ij}} = \frac{A_{ij}}{\sqrt{deg(i)deg(j)}}, (i \neq j)
\end{equation} 

\noindent where $A$ and $\tilde{A}$ are the adjacent matrix and negative normalized graph Laplacian matrix, and $deg(i)$ means the degree of $v_i$. It should be noted that our equation is slightly different from \cite{zhou2004learning}'s because we only consider explicit predicted label rather than label distribution.

Normalization is the core of reducing the marginal nodes' influence. Our approach is  simple: generating fake nodes, linking them to nearest real nodes and solving graph Laplacian regularization. \emph{Fake} label is not allowed to be assigned to unlabeled nodes and loss computation only considers edges between real nodes. The only difference between before and after generation is that marginal nodes' degree changes. And then the regularization parameter $\alpha_{ij}$ changes. 

\subsection{Proof}
We analyze 
how generated fake samples help acquire correct classification.

\begin{corollary}\label{decrease} Under Assumption \ref{convergence}, let $\mathcal{L}(\mathcal{C}_{gt})$ and $\mathcal{L}(\mathcal{C}_{gt})'$ be losses of ground truth on graph $(V + V_g, E + E_g)$ and $(V + V_g', E + E_g')$. We have $\forall V_g \supsetneqq V_g'$, $\mathcal{L}(\mathcal{C}_{gt}) < \mathcal{L}(\mathcal{C}_{gt})'$, where $V_g$ and $E_g$ are set of generated nodes and edges.
\end{corollary}

Corollary \ref{decrease} can be easily deduced because of $\alpha_{ij}$ decreasing. Loss of ground truth continues to decrease along with new fake samples being generated. That indicates ground truth is more likely to be acquired. However, there might exist other classification solutions whose loss decreases more. Thus, we will further prove that we can make a perfect classification under reasonable assumptions with adequate generated samples.
 
\begin{definition} \textbf{Partial Graph.} We define the subgraph induced by all nodes labeled $c$ (aka. $V_c$) and their other neighbors $Ne_c$ as partial graph $G_c$.
\end{definition}
\begin{assumption}\label{connectivity} \textbf{Connectivity.}
The subgraph induced by all interior nodes in each class is connected. Besides, every marginal node connects to at least one interior node in the same class.
\end{assumption}

\hide{
\begin{assumption}\label{necessary} \textbf{Necessary Labels.} 
There is at least one labeled node in every class.
\end{assumption}
}

Most real-world networks are dense and big enough to satisfy Assumption \ref{connectivity}. 
There actually implies another weak assumption that at least one labeled node exists for each class. This is the usually guaranteed by the setting of semi-supervised learning.
Let $m_c$ be the number of edges between marginal nodes in $G_c$. Besides, we define $deg_c$ as the maximum of degrees of nodes in $G_c$ and $loss_i$ as the supervised loss for misclassified labeled node $v_i$.
	
\begin{theorem} \textbf{Perfect Classification.} 
If enough fake samples are generated such that $\forall v \in \mathcal{M}, deg(v) > d_0$, all nodes will be correctly classified. $d_0$ is the maximum of $\max\limits_cm_c^2deg_c$  and $\max\limits_{c,v_i\in V_c}\frac{\lambda m_c}{loss_i}$.
\end{theorem}
\begin{proof}
We firstly consider a simplified problem in partial graphs $G_c$, where nodes from $Ne_c$ have already been assigned fixed label $c'$. We will prove that the new optimal classification $\mathcal{C}_{min}$ are the classification $\mathcal{C}'$, which correctly assigns $V_c$ label $c$. Since $\mathcal{L}(\mathcal{C}_{min})<\mathcal{L}(\mathcal{C}') < \lambda m_c\cdot \frac{1}{\sqrt{d_0\cdot d_0}} < \lambda m_c / \max\limits_{c,v_i\in V_c}\frac{\lambda m_c}{loss_i} \leq \min\limits_{v_i\in V_c}loss_i$, optimal solution $\mathcal{C}_{min}$ should classify all labeled nodes correctly. 

Suppose that $\mathcal{C}_{min}$ assigns $v_i,v_j\in \mathcal{I}$ with different labels. The inequality $\mathcal{L}(\mathcal{C}_{min})\geq \lambda\alpha_{ij} = \frac{\lambda}{\sqrt{deg(v_i)deg(v_j)}}\geq \frac{\lambda}{deg_c} \geq \frac{\lambda}{m_cdeg_c} \geq \frac{\lambda m_c}{d_0} > loss_{\mathcal{C}'}$ would result in contradiction.
According to analysis above and Assumption \ref{connectivity}, all interior nodes in $G_c$ are assigned label a $c$ in $\mathcal{C}_{min}$.

Suppose that $\mathcal{C}_{min}$ assigns $v_i \in \mathcal{M}\cap V_c, v_j\in \mathcal{I}$ with different labels and $(v_i, v_j)\in E$. Let $v_i$  be assigned with $c'$. If we change $v_i$'s label to $c$, then $\alpha_{ij}$ between $v_i$ and its interior neighbors will be excluded from the loss function. But some other edges weights between $v_i$ and its marginal neighbors might be added to the loss function. Let $\lambda\Delta$ denotes the variation of loss. 
The following equation will show that the decrease of the loss would lead to a contradiction.

\begin{equation*}
\footnotesize
\begin{split}
\Delta \leq & \sum\limits_{\substack{v_k\in \mathcal{M}\\(v_i, v_k) \in E_c}}\frac{1}{\sqrt{deg(v_i)deg(v_k)}} -  \sum\limits_{\substack{v_j\in \mathcal{I}\\(v_i, v_j) \in E_c}}\frac{1}{\sqrt{deg(v_i)deg(v_j)}} \\ 
\leq & \frac{1}{\sqrt{deg(v_i)}} (\frac{m_c}{\sqrt{\max\limits_cm_c^2deg_c}} - \sum\limits_{\substack{v_j\in \mathcal{I}\\(v_i, v_j) \in E_c}}\frac{1}{\sqrt{deg(v_j)}}) < 0
\end{split}
\normalsize
\end{equation*}

Suppose that $\mathcal{C}_{min}$ avoids all situations discussed above while $v_i \in \mathcal{M}\cap V_c$ is still assigned with $c'$. Under Assumption \ref{connectivity}, there exists an interior node $v_j$ connecting with $v_i$. As we discussed, $v_j$ must be assigned $c$ in $\mathcal{C}_{min}$, leading to contradiction. 
Therefore, $\mathcal{C}'$ is the only choice for optimal binary classification in $G_c$. That means all nodes in class $c$ are classified correctly. But what if in $G_c$ not all nodes in $Ne_c$ are labeled $c'$? 
Actually no matter which labels they are assigned, all nodes in $V_c$ are classified correctly. If nodes in $Ne_c$ are assigned labels except $c$ and $c'$, the proof is almost identical and $\mathcal{C}'$ is still optimal. If any nodes in $Ne_c$ are mistakenly assigned with label $c$, the only result is to encourage nodes to be classified as $c$ correctly. 

Finally, the analysis is correct for all classes thus all nodes will be correctly classified.
\end{proof}

\section{Experiments}\label{sec:experiments}

We conduct experiments on two citations networks and one entity extraction dataset.
Table \ref{table:dataset} summaries statistics of the three datasets.

To avoid over-tuning the network architectures and hyperparameters, in all our experiments we use a default settings for training and test.
Specifically, the classifier has 5 hidden layers with $(500,500,250,250,250)$ units. Stochastic layers are zero-centered Gaussian noise, with 0.05 standard deviation for input and 0.5 for outputs of  hidden layers. Generator has two 500-units hidden layers, each followed by a {\it batch normalization} layer. {\it Exponential Linear Unit} (ELU)~\cite{clevert2015fast} is used for improving the learning accuracy 
except the output layer of $G$, which instead uses $tanh$ to generate samples ranging from $-1$ to $1$. The trade-off factors in Algorithm~\ref{algo} are $\lambda_0=2,\lambda_1=1,\lambda_2=0.3$. Models are optimized by Adam~\cite{kingma2014adam}, where $\beta_1=0.5,\beta_2=0.999$. All parameters are initialized with Xavier~\cite{glorot2010understanding} initializer. 

\begin{table}[b]
 \setlength{\tabcolsep}{4pt}
\caption{Dataset statistics}\label{table:dataset}
\centering
\begin{tabular}{cccccc}
\hline
Dataset&nodes&edges&features&classes&labeled data\\
\hline
Cora & 2,708 & 5,429 & 1,433 & 7 & 140 \\
Citeseer& 3,327 & 4,732 & 3,703 & 6 & 120 \\
DIEL & 4,373,008 & 4,464,261 & 1,233,597 & 4 & 3413.8 \\
\hline
\end{tabular}
\end{table}

\subsection{Results on Citation Networks}\label{citation}
The two citation networks contain papers and citation links between papers.
Each paper has features represented as a bag-of-words and belongs to a specific class based on topic, such as ``database'' or ``machine learning''.
The goal is to classify all papers into its correct class.
For fair comparison, we follow exactly the experimental setting in \cite{yang2016revisiting}, where for each class 20 random instances (papers) are selected as labeled data and 1,000 instances as test data. The reported performance is the average of ten random splits.
In both datasets, we compare our proposed methods with three categories of methods:
\begin{itemize}
	\item \textbf{regularization-based methods} \\ including LP~\cite{zhu2002learning}, ICA~\cite{lu2003link}, and ManiReg~\cite{belkin2006manifold};
	\item \textbf{embedding-based methods} \\ including DeepWalk~\cite{Perozzi:14KDD}, SemiEmb~\cite{weston2012deep}, and Planetoid~\cite{yang2016revisiting};
	\item and \textbf{convolution-based methods}\\ including Chebyshev~\cite{defferrard2016convolutional}, GCN~\cite{kipf2016semi} and GAT~\cite{DBLP:journals/corr/abs-1710-10903}.
\end{itemize}
    
We train our models using {\it early stop} with 500 nodes for validation (average 20 epochs on Cora and 35 epochs on Citeseer). Every epoch contains 100 batches with batch size 64.
Table~\ref{tb:citation} shows the results of all comparison methods. Our method significantly outperforms all the regularization- and embedding-based methods, and also performs much better than Chebyshev and graph convolution networks (GCN), meanwhile slightly better than GCN with attentions (GAT). Compared with convolution-based methods, \smodel is more sensitive to labeled data and thus a larger variance is observed. The large variance might originate from the instability of training of GANs. For example, the \textit{mode collapse} phenomenon~\cite{theis2016note} will hamper \smodel from generating fake nodes evenly in density gaps. 
These instabilities are currently main problems in research of GANs. More advanced techniques for stabilizing \smodel are left for future work.

\begin{table}[t]
\caption{\label{tb:citation}Summary of results of classification accuracy (\%).}\label{results}
\centering
\begin{tabular}{c|ccc}
\hline \hline
Category&Method & Cora & Citeseer\\
\hline
\multirow{3}{*}{Regularization} & LP & 68.0 & 45.3\\
& ICA & 75.1 & 69.1\\
& ManiReg & 59.5 & 60.1 \\
\hline
\multirow{3}{*}{Embedding} & DeepWalk & 67.2 & 43.2\\ 
& SemiEmb & 59.0 & 59.6\\
& Planetoid & 75.7 & 64.7\\ 
\hline
\multirow{3}{*}{Convolution} 
& Chebyshev & 81.2 & 69.8\\ 
&GCN & 80.1 $\pm$ 0.5 & 67.9 $\pm$ 0.5\\
& GAT & \textbf{83.0 $\pm$ 0.7} & 72.5 $\pm$ 0.7\\
\hline
Our Method &\model & \textbf{83.0 $\pm$ 1.3} & \textbf{73.1 $\pm$ 1.8}\\
\hline \hline
\end{tabular}
\end{table}

\subsection{Verification}
\label{sec:verify}

\begin{figure*}[htbp]
 \centering
  \includegraphics[width=1.0\textwidth]{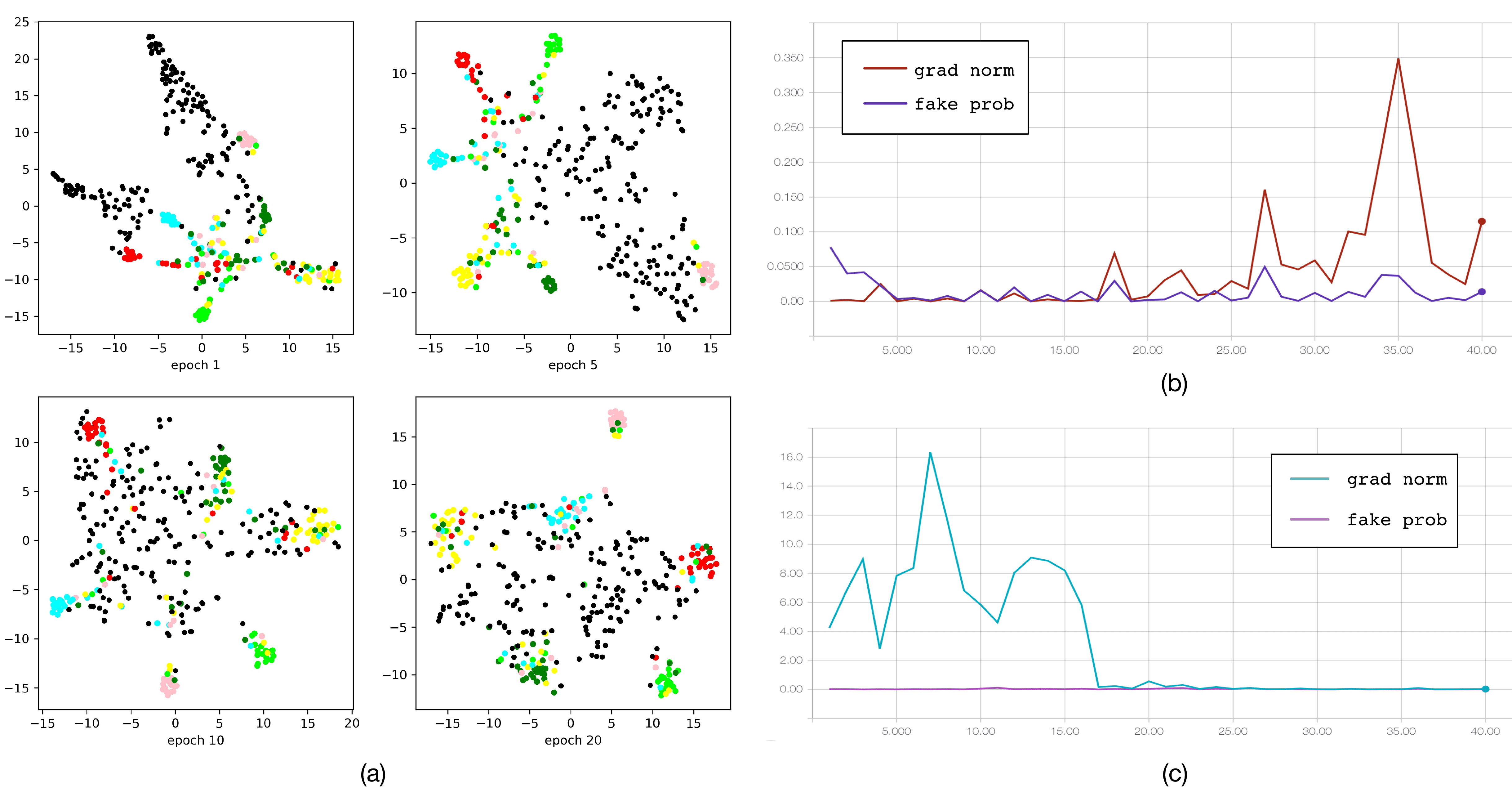}
 \caption{(a)Visualization of outputs of feature layer during training. Color of point indicates its class, where black nodes indicate fake samples. 
 	 (b)Typical $||g||$ and $p_f$ curve for marginal nodes. Horizontal axis represents the number of training iterations, vertical axis representing the value. (c) Typical $||g||$ and $p_f$ curve for interior nodes.}\label{process}
 \end{figure*}

We provide more insights into \smodel with experimental verifications. There are two verification experiments: one is about the expected equilibrium, and the other verifies the working principles of \model.

\subsubsection{Verification of equilibrium}
The first experiment 
is about whether \smodel converges at the equilibrium described in \S~\ref{LD}. In Figure \ref{process}(a), we visualize the training process in the Citeseer experiment using t-SNE algorithm (Cf. \S~\ref{sec:experiments} for detailed experimental settings). At the beginning, $G$ generates samples very different from real samples and the boundaries of clusters are ambiguous. During training, classifier $D$ gradually learns a non-linear transformation in $h^{(n)}(\mathbf{x})$ to map real and fake samples into distinct clusters, while $G$ tries to generate samples in the central areas of real samples. Mini-batch training, $loss_{pt}$ in $\mathcal{L}_G$ and Gaussian noises prevent $G$ from overfitting on the only center point. Adversarial training finally reaches the expected equilibrium where fake samples are in the central areas surrounded by real samples clusters after 20 epochs. 

\subsubsection{Verification of working principles}

The second experiment is to verify the proposed working principles. We have proved in \S~\ref{sec:theory} theoretically that reducing the influence of marginal nodes can help classification. But we should further verify whether generated samples reduce the influence of marginal nodes. On one hand, nodes are mapped into distinct and far clusters in $h^{(n)}(\mathbf{x})$. On the other hand, the ``influence'' is related with ``smooth degree''. For example in graph Laplacian regularization framework, difference between labels of adjacent nodes are minimized explicitly to guarantee the smoothness. 
Thus we examine classifier function's smooth degree around density gaps. Smooth degrees at $x_i$ are measured  by the norm of the gradient of maximum in probabilities for each class. 

	\[||g(\mathbf{x}_i)|| = ||\nabla_{\mathbf{x}_i} \max\limits_{y=0}^{M-1} P(y |\mathbf{x}_i, y_i < M)||\]

Let $p_f(\mathbf{x}_i)$ be predicted fake probability of $\mathbf{x}_i$. We draw curves of $||g(\mathbf{x}_i)||$ and $p_f(\mathbf{x}_i)$ during training. Figure~\ref{process}(b)(c) show two representative patterns: (b) is a marginal node, whose $||g(\mathbf{x}_i)||$ and $p_f(\mathbf{x}_i)$ change synchronously and strictly share the same trend, while (c) is an interior node never predicted fake. The classifier function around (c) remains smooth after determining a definite label. Pearson correlation coefficient $$r_p = \frac{cov(||g||, p_f)}{\sigma_{||g||}\sigma_{p_f}}$$ exceeds 0.6, indicating obvious positive correlation.

\subsection{Results on Entity Extraction}
\label{diel}
The DIEL dataset~\cite{bing2015improving} is a dataset for information extraction. It contains pre-extracted features for each entity mentions in text, and a graph connecting entity mentions to corresponding coordinate-item lists. The objective is to extract medical entities from items given feature vectors, the graph topologies and a few known medical entities.

Again, we follow the same setting as in the original paper~\cite{bing2015improving} for the purpose of comparison, including data split and the average of different runs.
Because the features are very high dimensional sparse vectors, we reduce its dimensions to 128 by Truncated SVD algorithm. We use neighbor fusion on {\it item string nodes} with $\alpha = 0$ as only {\it entity mention nodes} have features. 
We treat the top-$k$ ($k$ = 240,000) entities given by a model as positive, and compare 
recall of top-$k$ ranked results by different methods.
Note that as many items in ground truth do not appear in text, the upper bound of recall is 0.617 in this dataset.

We also compare it with the different types of methods. 
Table~\ref{diel_result}
reports the average recall@$k$ of standard data splits for 10 runs by all the comparison methods.
{\it DIEL} 
represents the method in original paper~\cite{bing2015improving}, which uses outputs of multi-class label propagation to train a classifier. 
The result of Planetoid is the inductive result which shows the best performance among three versions. As the DIEL dataset has millions of nodes and edges, which makes full-batch training, for example GCN, infeasible(using sparse storage, memory needs 
$>$ 200GB),  we do not report GCN and GAT here. From Table~\ref{diel_result}, we see that our method \smodel achieves the best performance, significantly outperforming all the comparison methods ($p-$value$\ll$0.01, $t-$test).

\subsection{Space Efficiency}

Since GCN cannot handle large-scale networks, we examine the memory consumption of \smodel in practice. GPU has become the standard platform for deep learning algorithms. The memory on GPU are usually very limited compared with the main memory. We compare the GPU memory consumption of four representative algorithms from different categories in Figure \ref{bar}. Label Propagation does not need GPU, we show its result on CPU for the purpose of comparison. 

For small dataset, \smodel consumes the largest space due to the most complex structure. But for large dataset, \smodel uses the least GPU memories. LP is usually implemented by solving equations, whose space complexity is $O(N^2)$. Here we use the ``propagation implementation'' to save space. GCN needs full-batch training and cannot handle a graph with millions of nodes. Planetoid and \smodel are trained in mini-batch way so that the space consumption is independent of the number of nodes. High dimensional features in DIEL dataset are also challenging. Planetoid uses sparse storage to handle sparse features and \smodel reduces the dimension using Truncated SVD.

\subsection{Adversarial Label Propagation}
 \label{subsec:morediscussions}

Will adversarial learning help conventional semi-supervised learning algorithms?
Yes, we have proved theoretically that reducing the influence of marginal nodes can help classification in \S~\ref{sec:theory}. So, we further propose an adversarial improvement for graph Laplacian regularization with generated samples to verify our proofs (Cf. \S~\ref{sec:theory}).
We conduct the experiment by incorporating adversarial learning into the
Label Propagation framework~\cite{zhou2004learning} to see whether the performance can be improved or not.
To incorporate nodes' features, we reconstruct graph by linking nodes to their $k$ nearest neighbors in feature space $\mathbf{x}$. We use the Citeseer network in this experiment, whose settings are described in \S~\ref{citation} meanwhile $k=10$.
 Generating enough fake data is time-consumable, therefore we directly determined marginal nodes by $p_f(\mathbf{x}_i) >\tau $ ($\tau$ is a threshold), because $$\mathbb{E}\frac{p_g(\mathbf{x})}{p_g(\mathbf{x}) + p_{data}(\mathbf{x})} =p_f(\mathbf{x})$$ 
\noindent where $p_f(\mathbf{x})$ is the probability of $\mathbf{x}$ to be fake.

Then, we increase marginal nodes' degree up to $r$ to reduce their influence. 

Figure~\ref{adv} shows the results. Smaller $p_f$ means more marginal nodes. The curves indicates that $\tau = 0.1$ is a good threshold to discriminate marginal and interior nodes. 
A bigger $r$ always perform better, which encourages us to generate as much fake nodes in density gaps as possible.
The performance of LP has been improved by increasing marginal nodes' degree, but still underperforms \model. 

In our opinion, the main reason is that \smodel uses neural networks to capture high-order interrelation between features. 
Thus, we also try to reconstruct the graph using the first layer's output of a supervised multi-layer perceptron and further observed improvements in performance, highlighting the power of neural networks in this problem. 

\begin{table}
\centering
\caption{Recall@k on DIEL dataset (\%).}\label{diel_result}
\begin{tabular}{c|cc}
\hline \hline
Category&Method & Recall@K\\
\hline
\multirow{2}{*}{Regularization} & LP & 16.2\\
&ManiReg & 47.7 \\
\hline
\multirow{3}{*}{Embedding}& DeepWalk & 25.8 \\
&SemiEmb & 48.6 \\
&Planetoid & 50.1 \\
\hline
Original&\textit{DIEL} & 40.5\\
\hline
Our Method& \model & \textbf{51.8} \\
\hline
&Upper bound & 61.7\\
\hline \hline
\end{tabular}
\end{table}

\begin{figure}
	\centering
	\includegraphics[width=0.85\linewidth]{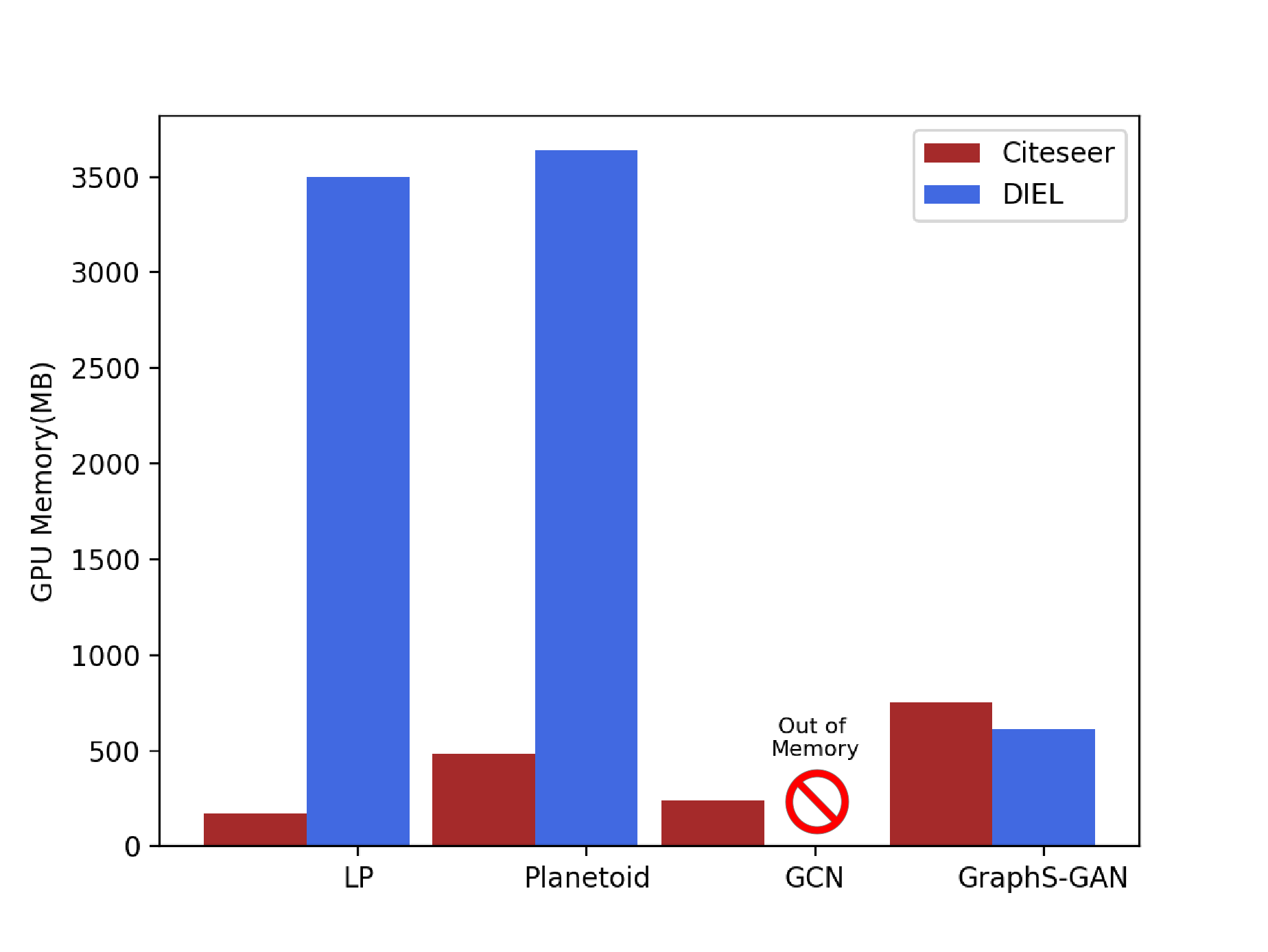}
	\caption{GPU Memory Consumption of four typical semi-supervised learning algorithms on graphs. GCN cannot handle large-scale networks. Planetoid and \smodel are trained in mini-batch way, which makes them able to scale up. We reduce the dimensions of features in DIEL dataset(Cf. \S~\ref{diel}) so that the GPU consumption of \smodel even decreases.} \label{bar}
\end{figure}

\begin{figure}
	\centering
	\includegraphics[width=0.9\linewidth]{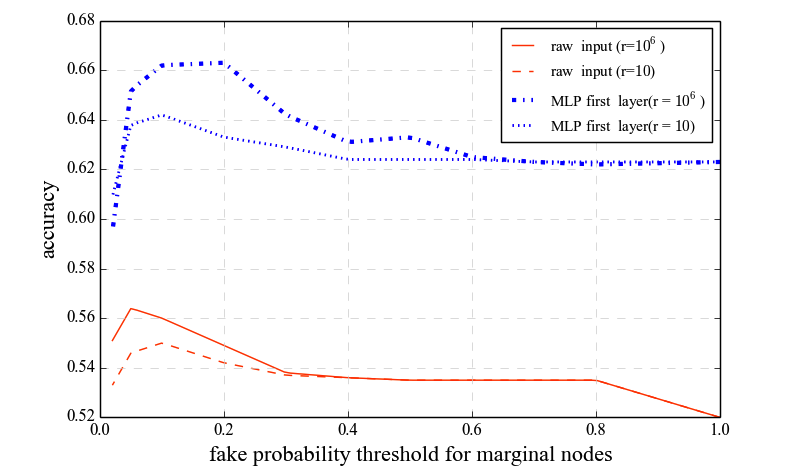}
	\caption{Performance of LP with adversarial learning. Larger threshold means less marginal nodes. Two curves below take $x_i$ as inputs and curves above use outputs of the first layer of MLP. 
	} \label{adv}
\end{figure}

\section{Related Work}\label{sec:related}
Related works mainly fall into three categories: Algorithms for semi-supervised learning on graphs, GANs for semi-supervised learning and GAN-based applications on graphs. We discuss them and summarize the main differences between our proposed model and these works as follows:
\subsection{Semi-supervised Learning on Graphs}
	As mentioned in \S~\ref{sec:intro}, previous methods for this task can be divided into three categories. 
    
Label Propagation~\cite{zhu2002learning} is the first work under the graph Laplacian framework. Labeled nodes continues to propagate their labels to adjacent nodes until convergence. After revealing the relationship between LP and graph Laplacian regularization~\cite{zhu2003semi}, the method are improved by sophisticated smoothing regularizations~\cite{zhu2003semi,belkin2006manifold} and bootstrap method~\cite{lu2003link}. This kind of methods mainly focus on local smoothness but neglect clustering property of graphs, making situations like Figure~\ref{example} hard cases. 
        
Deepwalk~\cite{Perozzi:14KDD} is the first work for graph embedding. As an unsupervised method to learn latent representations for nodes, DeepWalk can easily be turned to a semi-supervised baseline model if combined with SVM classifier. Since labels help learn embeddings and then help classification, Planetoid~\cite{yang2016revisiting} jointly learns graph embeddings and predicts node labels.
Graph embedding becomes one step in \smodel and we incorporate GANs for better performance.  

GCN~\cite{kipf2016semi} is the first graph convolution model for semi-supervised learning on graphs. Every filter in GCN learns linear transformations on spectral domain for every feature and combines them. More complex graph convolution methods~\cite{defferrard2016convolutional,DBLP:journals/corr/abs-1710-10903} show better performances. An obvious disadvantage of graph convolution is huge consumptions of space, which is overcome by \model. 

\subsection{GANs for Semi-supervised Learning}

Semi-supervised GANs(SGAN) were first put forward in computer vision domain~\cite{odena2016semi}. SGAN just replaces the discriminator in GANs with a classifier and becomes competitive with state-of-art semi-supervised models for image classification. Feature matching loss is first put forward to prevent generator from overtraining~\cite{salimans2016improved}. The technique is found helpful for semi-supervised learning, leaving the working principles unexplored~\cite{salimans2016improved}. Analysis on the trade-off between the classification performance of semi-supervised and the quality of generator was given in~\cite{dai2017good}. Kumar et al.~\shortcite{kumar2017semi} find a smoothing method by estimating the tangent space to the data manifold. In addition, various auxiliary architectures are combined with semi-supervised GANs to classify images more accurately~\cite{maaloe2016auxiliary,dumoulin2016adversarially,chongxuan2017triple}. All these works focus on image data and leverage CNN architectures. \smodel introduces this thought to graph data and first designs a new GAN-like game with clear and convincing working principles. 

\subsection{GAN-based Applications on Graphs}
Although we firstly introduce GANs to graph-based semi-supervised learning problem, GANs have made successes in many other machine learning problems on graphs.

One category is about graph generation. Liu et al.~\shortcite{liu2017learning} present a hierarchical architecture composed by multiple GANs to generate graphs. The model preserves topological features of training graphs. Tavakoli et al.~\shortcite{tavakolilearning} apply GANs for link formation in social networks. Generated network preserves the distribution of links with minimal risk of privacy breaches. 

Another category is about graph embedding. 
In GraphGAN~\cite{wang2017graphgan}, generator learns embeddings for nodes and discriminator solves link prediction task based on embeddings. Classification-oriented embeddings are got at the equilibrium. Dai et al.~\shortcite{dai2017adversarial} leveraged adversarial learning to regularize training of graph representations. Generator transforms embeddings from traditional algorithms into new embeddings, which not only preserve structure information but also mimic a prior distribution.

\section{Conclusion and Future Work}
\label{sec:discuss}

We propose \model, a novel approach for semi-supervised learning over graphs using GANs. We design a new competitive game between generator and classifier, in which generator generates samples in density gaps at equilibrium. Several sophisticated loss terms together guarantee the expected equilibrium.
Experiments on three benchmark datasets demonstrate the effectiveness of our approach. 

We also provide a thorough analysis of working principles behind the proposed model \model. Generated samples reduce the influence of nearby nodes in density gaps so as to make decision boundaries clear. The principles can be generalized to improve traditional algorithms based on graph Laplacian regularization with theoretical guarantees and experimental validation. 
\smodel is scalable. Experiments on DIEL dataset suggest that our model shows good performance on large graphs too.

As future work, one potential direction is to
investigate more ideal equilibrium, stabilizing the training further, accelerating training, strengthening theoretical basis of this method and extending the method to other tasks on graph data such as~\cite{Tang:08KDD}. 

\vpara{Acknowledgements.}The work is supported by the
(2015AA124102),
National Natural Science Foundation of China (61631013,61561130160),
and 
the Royal Society-Newton Advanced Fellowship Award.

\bibliographystyle{ACM-Reference-Format}
\bibliography{ref}

\end{document}